\documentclass[a4paper]{easychair}

\usepackage{latexsym}
\usepackage{amssymb}
\usepackage{pstricks}
\psset{xunit=1mm,yunit=1mm,runit=1mm,linearc=1.5,linewidth=0.5pt}
\newtheorem{theorem}{Theorem}
\newtheorem{definition}[theorem]{Definition}
\newtheorem{proposition}[theorem]{Proposition}
\newtheorem{lemma}[theorem]{Lemma}
\newtheorem{corollary}[theorem]{Corollary}
\newcommand{\eh}{\underline h_T}
\newcommand{\hh}{\overline h_T}
\newcommand{\hhc}[1]{\overline h_{T_{#1}}}
\newcommand{\ed}{\underline d_T}
\newcommand{\hd}{\overline d_T}
\newcommand{\hdc}[1]{\overline d_{T_{#1}}}
\newcommand{\EH}{\underline H_T}
\newcommand{\HH}{\overline H_T}
\newcommand{\ED}{\underline D_T}
\newcommand{\HD}{\overline D_T}
\newcommand{\nl}{\begin{pspicture}(2.3,0)
\psline[linearc=0]{->}(2,1.5)(2,.5)(0,.5)
\end{pspicture}}

\title{Asymptotic Proportion of Hard Instances of the Halting
Problem\footnote{\textbf{Version note.}
The differences between this version and arXiv:1307.7066v1 are significant.
They have been listed in the last paragraph of Section~\ref{S:intro}.
This publication has appeared in Acta Cybernetica \textbf{21} (2014) 307--330.
Excluding layout, this arXiv version is essentially identical to the Acta
Cybernetica version.
The author thanks the editor of Acta Cybernetica for the kind permission to
submit this version to arXiv.}
}
\author{Antti Valmari\\
Tampere University of Technology,
Department of Mathematics\\
PO Box 553, FI-33101 Tampere, FINLAND}

\begin{document}
\maketitle

\begin{abstract}
Although the halting problem is undecidable, imperfect testers that fail on
some instances are possible.
Such instances are called \emph{hard} for the tester.
One variant of imperfect testers replies ``I don't know'' on hard instances,
another variant fails to halt, and yet another replies incorrectly ``yes'' or
``no''.
Also the halting problem has three variants: does a given program halt on the
empty input, does a given program halt when given itself as its input, or does
a given program halt on a given input.
The failure rate of a tester for some size is the proportion of hard instances
among all instances of that size.
This publication investigates the behaviour of the failure rate as the size
grows without limit.
Earlier results are surveyed and new results are proven.
Some of them use C++ on Linux as the computational model.
It turns out that the behaviour is sensitive to the details of the programming
language or computational model, but in many cases it is possible to prove
that the proportion of hard instances does not vanish.
\end{abstract}

{\small

\noindent\textbf{ACM Computing Classification System 1998:}
F.1.1 Models of Computation--Computability theory

}

\section{Introduction}\label{S:intro}

Turing proved in 1936 that undecidability exists by showing that the halting
problem is undecidable~\cite{Tur36}.
Rice extended the set of known undecidable problems to cover all questions of
the form ``does the partial function computed by the given program have
property $X$'', where $X$ is any property that at least one computable partial
function has and at least one does not have~\cite{Ric53}.
For instance, $X$ could be ``returns $1$ for all syntactically correct C++
programs and $0$ for all remaining inputs.''
In other words, it may be impossible to find out whether a given weird-looking
program is a correct C++ syntax checker.
These results are basic material in such textbooks as~\cite{HoU79}.

On the other hand, imperfect halting testers are possible.
For any instance of the halting problem, a \emph{three-way tester} eventually
answers ``yes'', ``no'', or ``I don't know''.
If it answers ``yes'' or ``no'', then it must be correct.
We say that the ``I don't know'' instances are \emph{hard instances} for the
tester.
Also other kinds of imperfect testers have been introduced, as will be
discussed in Section~\ref{S:variants}.

Assume that $T_1$ is a tester.
By Turing's proof, it has a hard instance $I_1$.
If $I_1$ is a halting instance, then let $T_2$ be ``if the input is $I_1$,
then reply `yes', otherwise run $T_1$ and return its reply''.
If $I_1$ is non-halting, then let $T_2$ be ``if the input is $I_1$, then reply
`no', otherwise run $T_1$ and return its reply''.
By construction, $T_2$ is a tester with one fewer hard instances than $T_1$
has.
By Turing's proof, also $T_2$ has a hard instance.
Let us call it $I_2$.
It is hard also for $T_1$.
This reasoning can be repeated without limit, yielding an infinite sequence
$T_1$, $T_2$, \ldots\ of testers and $I_1$, $I_2$, \ldots\ of instances such
that $I_i$ is hard for $T_1$, \ldots, $T_i$ but not for $T_{i+1}$, \ldots.
Therefore, every tester has an infinite number of hard instances, but no
instance is hard for all testers.

A program that answers ``I don't know'' for every program and input is a
three-way tester, although it is useless.
A much more careful tester simulates the given program on the given input at
most $9^{9^n}$ steps, where $n$ is the joint size of the program and its
input.
If the program stops by then, then the tester answers ``yes''.
If the program repeats a configuration (that is, a complete description of the
values of variables, the program counter, etc.) by then, then the tester
answers ``no''.
Otherwise it answers ``I don't know''.
With this theoretically possible but in practice unrealistic tester, any hard
halting instance has a finite but very long running time.

The proofs by Turing and Rice may leave the hope that only rare artificial
contrived programs yield hard instances.
One could dream of a three-way tester that answers very seldom ``I don't
know''.
This publication analyses this issue, by surveying and proving results that
tell how the proportion of hard instances behaves when the size of the
instances grows without limit.

Section~\ref{S:def} presents the variants of the halting problem and imperfect
testers surveyed, together with some basic results and notation.
Earlier research is discussed in Section~\ref{S:related}.
The section contains some proofs to bring results into the framework of this
publication.
Section~\ref{S:proglang} presents some new results in the case that a program
has many copies of all big sizes, or information can be packed densely inside
the program.
It is not always assumed that the program has access to the information.
A natural example of such information is dead code, such as
\texttt{if(1==0)then\{\ldots\}}.
In Section~\ref{S:C++}, results are derived for C++ programs with inputs from
files.
Section~\ref{S:conclusions} briefly concludes this publication.

This publication is a significantly extended version of~\cite{Val13a,Val13b}.
The papers~\cite{Val13a,Val13b} are otherwise essentially the same, but three
proofs were left out from~\cite{Val13b} because of lack of space.
In the present publication, Theorems~\ref{T:ill-eof} and~\ref{T:hard-sd} and
Corollaries~\ref{C:nolimit1} and~\ref{C:nolimit2} are new results lacking
from~\cite{Val13a,Val13b}.
Furthermore,~\cite{Val13a,Val13b} incorrectly claimed the opposite of
Theorem~\ref{T:hard-sd}.
The present publication fixes this error and also a small error in
Proposition~\ref{P:syntax}.

\section{Concepts and Notation}\label{S:def}

\subsection{Variants of the Halting Problem}\label{S:variants}

The literature on hard instances of the halting problem considers at least
three variants of the halting problem:
\begin{description}
\item[E] does the given program halt on the \emph{empty} input~\cite{HaM06},
\item[S] does the given program halt when given \emph{itself} as its
input~\cite{Lyn74,Ryb07}, and
\item[G] does the given program halt on the \emph{given}
input~\cite{CaS08,KSZ05,ScJ99}.
\end{description}
Each variant is undecidable.
Variant G has a different notion of instances from others: program--input
pairs instead of just programs.
A tester for G can be trivially converted to a tester for E or S, but the
proportion of hard program--input pairs among all program--input pairs of some
size is not necessarily the same as the similar proportion with the input
fixed to the empty one or to the program itself.

The literature also varies on what the tester does when it fails.
Three-way testers, that is, the ``I don't know'' answer is used implicitly
by~\cite{Lyn74}, as it discusses the union of two decidable sets, one being a
subset of the halting and the other of the non-halting instances.
In \emph{generic-case decidability}~\cite{Ryb07}, instead of the ``I don't
know'' answer, the tester itself fails to halt.
Yet another idea is to always give a ``yes'' or ``no'' answer, but let the
answer be incorrect for some instances~\cite{KSZ05,ScJ99}.
Such a tester is called \emph{approximating}.
One-sided results, where the answer is either ``yes'' or ``I don't know'',
were presented in~\cite{CaS08,HaM06}.
For a tester of any of the three variants, we say that an instance is
\emph{easy} if the tester correctly answers ``yes'' or ``no'' on it, otherwise
the instance is \emph{hard}.

\newcommand{\threeway}[1]{\textnormal{three-way(#1)}}
\newcommand{\generic}[1]{\textnormal{generic(#1)}}
\newcommand{\apprx}[1]{\textnormal{approx(#1)}}
These yield altogether nine different sets of testers, which we will denote
with \threeway{X},\linebreak
\generic{X}, and \apprx{X}, where X is E, S, or G.
Some simple facts facilitate carrying some results from one variant of testers
to another.

\begin{proposition}\label{P:3way->}
For any three-way tester there is a generic-case tester that has precisely the
same easy ``yes''-instances, easy ``no''-instances, hard halting instances,
and hard non-halting instances.

There also is an approximating tester that has precisely the same easy
``yes''-instances, at least the same easy ``no''-instances, precisely the same
hard halting instances, and no hard non-halting instances; and an
approximating tester that has at least the same easy ``yes''-instances,
precisely the same easy ``no''-instances, no hard halting instances, and
precisely the same hard non-halting instances.
\end{proposition}
\begin{proof}
A three-way tester can be trivially converted to the promised tester by
replacing the ``I don't know'' answer with an eternal loop, the reply ``no'',
or the reply ``yes''.
\end{proof}

\begin{proposition}\label{P:gen->gen}
For any generic-case tester there is a generic-case tester that has at least
the same ``yes''-instances, precisely the same ``no''-instances, no hard
halting instances, and precisely the same hard non-halting instances.
\end{proposition}
\begin{proof}
In parallel with the original tester, the instance is simulated.
(In Turing machine terminology, parallel simulation is called
``dovetailing''.)
If the original tester replies something, the simulation is aborted.
If the simulation halts, the original tester is aborted and the reply ``yes''
is returned.
\end{proof}

\begin{proposition}\label{P:finite}
For any $i \in \mathbb{N}$ and tester $T$, there is a tester $T_i$ that
answers correctly ``yes'' or ``no'' for all instances of size at most $i$, and
similarly to $T$ for bigger instances.
\end{proposition}
\begin{proof}
Because there are only finitely many instances of size at most $i$, there is a
finite bit string that lists the correct answers for them.
If $n \leq i$, $T_i$ picks the answer from it and otherwise calls $T$.
(We do not necessarily know what bit string is the right one, but that does
not rule out its existence.)
\end{proof}

\subsection{Notation}\label{S:notation}

We use $\Sigma$ to denote the set of characters that are used for writing
programs and their inputs.
It is finite and has at least two elements.
There are $|\Sigma|^n$ character strings of size $n$.
If $\alpha$ and $\beta$ are in $\Sigma^*$, then $\alpha \sqsubseteq \beta$
denotes that $\alpha$ is a prefix of $\beta$, and $\alpha \sqsubset \beta$
denotes proper prefix.
The size of $\alpha$ is denoted with $|\alpha|$.

A set $A$ of finite character strings is \emph{self-delimiting} if and only if
membership in $A$ is decidable and no member of $A$ is a proper prefix of a
member of $A$.
The \emph{shortlex ordering} of any set of finite character strings is
obtained by sorting the strings in the set primarily according to their sizes
and strings of the same size in the lexicographic order.

Not necessarily all elements of $\Sigma^*$ are programs.
The set of programs is denoted with $\Pi$, and the set of all (not necessarily
proper) prefixes of programs with $\Gamma$.
So $\Pi \subseteq \Gamma$.
For tester variants E and S, we use $p(n)$ to denote the number of programs of
size $n$.
Then $p(n) = |\Sigma^n \cap \Pi|$.
For tester variant G, $p(n)$ denotes the number of program--input pairs of
joint size $n$.
We will later discuss how the program and its input are paired into a single
string.
The numbers of halting and non-halting (a.k.a. diverging) instances of size
$n$ are denoted with $h(n)$ and $d(n)$, respectively.
We have $p(n) = h(n) + d(n)$.

If $T$ is a tester, then $\eh(n)$, $\hh(n)$, $\ed(n)$, and $\hd(n)$ denote the
number of its easy halting, hard halting, easy non-halting, and hard
non-halting instances of size $n$, respectively.
Obviously $\eh(n) + \hh(n) = h(n)$ and $\ed(n) + \hd(n) = d(n)$.
The smaller $\hh(n)$ and $\hd(n)$ are, the better the tester is.
The \emph{failure rate} of $T$ is $(\hh(n) + \hd(n))/p(n)$.

When referring to all instances of size at most $n$, we use capital letters.
So, for example, $P(n) = \sum_{i=0}^n p(i)$ and $\HD(n) = \sum_{i=0}^n
\hd(i)$.

\section{Related Work}\label{S:related}

\subsection{Early Results by Lynch}

Nancy Lynch~\cite{Lyn74} used \emph{G\"odel numberings} for discussing
programs.
In essence, it means that each program has at least one index number (which is
a natural number) from which the program can be constructed, and each natural
number is the index of some program.

Although the index of an individual program may be smaller than the index of
some shorter program, the overall trend is that indices grow as the size of
the programs grows, because otherwise we would run out of small numbers.
On the other hand, if the mapping between the programs and indices is 1--1,
then the growth cannot be faster than exponential.
This is because $p(n) \leq |\Sigma|^n$.
With real-life programming languages, the growth is exponential, but (as we
will see in Section~\ref{S:C++model}) the base of the exponent may be smaller
than $|\Sigma|$.

To avoid confusion, we refrain from using the notation $\HH$, etc., when
discussing results in~\cite{Lyn74}, because the results use indices instead of
sizes of programs, and their relationship is not entirely straightforward.
Fortunately, some results of~\cite{Lyn74} can be immediately applied to
programming languages by using the \emph{shortlex G\"odel numbering}.
The shortlex G\"odel number of a program is its index in the shortlex ordering
of all programs.

The first group of results of~\cite{Lyn74} reveals that a wide variety of
situations may be obtained by spreading the indices of all programs sparsely
enough and then filling the gaps in a suitable way.
For instance, with one G\"odel numbering, for each three-way tester, the
proportion of hard instances among the first $i$ indices approaches $1$ as $i$
grows.
With another G\"odel numbering, there is a three-way tester such that the
proportion approaches $0$ as $i$ grows.
There even is a G\"odel numbering such that as $i$ grows, the proportion
oscillates in the following sense: for some three-way tester, it comes
arbitrarily close to $0$ infinitely often and for each three-way tester, it
comes arbitrarily close to $1$ infinitely often.

In its simplest form, spreading the indices is analogous to defining a new
language SpaciousC++ whose syntax is identical to that of C++ but the
semantics is different.
If the first $\lfloor n/2 \rfloor$ characters of a SpaciousC++ program of size
$n$ are space characters, then the program is executed like a C++ program,
otherwise it halts immediately.
This does not restrict the expressiveness of the language, because any C++
program can be converted to a similarly behaving SpaciousC++ program by adding
sufficiently many space characters to its front.
However, it makes the proportion of easily recognizable trivially halting
instances overwhelm.
A program that replies ``yes'' if there are fewer than $\lfloor n/2 \rfloor$
space characters at the front and ``I don't know'' otherwise, is a three-way
tester.
Its proportion of hard instances vanishes as the size of the program grows.

As a consequence of this and Proposition~\ref{P:finite}, one may choose any
failure rate above zero and there is a three-way tester for SpaciousC++
programs with at most that failure rate.
Of course, this result does not tell anything about how hard it is to test the
halting of interesting programs.
This is the first example in this publication of what we call \emph{an anomaly
stealing the result}.
That is, a proof of a theorem goes through for a reason that has little to do
with the phenomenon we are interested in.

Indeed, the first results of~\cite{Lyn74} depend on using unnatural G\"odel
numberings.
They do not tell what happens with untampered programming languages.
Even so, they rule out the possibility of a simple and powerful general
theorem that applies to all models of computation.
They also make it necessary to be careful with the assumptions that are made
about the programming language.

To get sharper results, \emph{optimal G\"odel numberings} were discussed
in~\cite{Lyn74}.
They do not allow distributing programs arbitrarily.
A G\"odel numbering is optimal if and only if for any G\"odel numbering, there
is a computable function that maps it to the former such that the index never
grows more than by a constant factor.\footnote
{The definition in~\cite{Lyn74} seems to say that the function must be a
bijection.
We believe that this is a misprint, because each proof in~\cite{Lyn74} that
uses optimal G\"odel numberings obviously violates it.}
The most interesting sharper results are opposite to what was obtained without
the optimality assumption.
To apply them to programming languages, we first define a programming language
version of optimal G\"odel numberings.

\begin{definition}\label{D:eof-data}
A programming language is \emph{end-of-file data segment}, if and only if each
program consists of two parts in the following way.
The first part, called the \emph{actual program}, is written in a
self-delimiting language (so its end can be detected).
The second part, called the \emph{data segment}, is an arbitrary character
string that extends to the end of the file.
The language has a construct via which the actual program can read the
contents of the data segment.
\end{definition}
The data segment is thus a data literal in the program, packed with maximum
density.
It is not the same thing as the input to the program.

\begin{corollary}\label{C:Lyn6}
For each end-of-file data segment language,
$$
\exists c > 0: \exists T \in \threeway{S}: \forall n \in \mathbb{N}:
\frac{\EH(n) + \ED(n)}{P(n)} \geq c\textrm{ and}\hspace{13mm}
$$
$$
\exists c > 0: \forall T \in \threeway{S}: \exists n_T \in \mathbb{N}:
\forall n \geq n_T: \frac{\HH(n) + \HD(n)}{P(n)} \geq c\textrm{ .}
$$
\end{corollary}
\begin{proof}
Let $\mathcal{L}$ be the end-of-file data segment language, and let
$\mathcal{G}$ be any G\"odel numbering.
Consider the following program $P$ in $\mathcal{L}$.
Let $a$ and $d$ be the sizes of its actual program and data segment.
The actual program reads the data segment, interpreting its content as a
number $i$ in the range from $\frac{|\Sigma|^d-1}{|\Sigma|-1}+1$ to
$\frac{|\Sigma|^{d+1}-1}{|\Sigma|-1}$.
Then it simulates the $i$th program in $\mathcal{G}$.
The shortlex index of $P$ is at most $i' = \sum_{j=0}^{a+d} |\Sigma|^j \leq
|\Sigma|^{a+d+1}$.
We have $\frac{|\Sigma|^d-1}{|\Sigma|-1} + 1 \leq i$, yielding $|\Sigma|^d - 1
\leq |\Sigma|i - i - |\Sigma| + 1$, so $|\Sigma|^d \leq |\Sigma|i$, thus $i'
\leq |\Sigma|^{a+2}i$.
The shortlex numbering of $\mathcal{L}$ is thus an optimal G\"odel numbering.
From this, Proposition~6 in~\cite{Lyn74} gives the claims.
\end{proof}
A remarkable feature of the latter result compared to many others in this
publication is that $c$ is chosen before $T$.
That is, there is a positive constant that only depends on the programming
language (and not on the choice of the tester) such that all testers have at
least that proportion of hard instances, for any big enough $n$.
On the other hand, the proof depends on the programming language allowing to
pack raw data very densely.
Real-life programming languages do not satisfy this assumption.
For instance, C++ string literals \texttt{"\ldots"} cannot pack data densely
enough, because the representation of \texttt{"} inside the literal (e.g.,
\texttt{\char`\\"} or \texttt{\char`\\042}) requires more than one character.

Because of Proposition~\ref{P:finite}, ``$\exists n_T \in \mathbb{N}$'' cannot
be moved to the front of ``$\forall T \in \threeway{S}$''.

The result cannot be generalized to $\hh$, $\hd$, and $p$, because the
following anomaly steals it.
We can change the language by first adding \texttt{1} or \texttt{01} to the
beginning of each program $\pi$ and then declaring that if the size of
\texttt{1}$\pi$ or \texttt{01}$\pi$ is odd, then it halts immediately,
otherwise it behaves like $\pi$.
This trick does not invalidate optimality but introduces infinitely many sizes
for which the proportion of hard instances is $0$.

\subsection{Results on Domain-Frequent Programming Languages}\label{S:KSZ05}

In~\cite{KSZ05}, the halting problem was analyzed in the context of
programming languages that are \emph{frequent} in the following sense:

\begin{definition}\label{D:frequent}
A programming language is (a) \emph{frequent} (b) \emph{domain-frequent}, if
and only if for every program $\pi$, there are $n_\pi \in \mathbb{N}$ and
$c_\pi > 0$ such that for every $n \geq n_\pi$, at least $c_\pi p(n)$ programs
of size $n$ (a) compute the same partial function as $\pi$ (b) halt on
precisely the same inputs as $\pi$.
\end{definition}
Instead of ``frequent'', the word ``dense'' was used in~\cite{KSZ05}, but we
renamed the concept because we felt ``dense'' a bit misleading.
The definition says that programs that compute the same partial function are
common.
However, the more common they are, the less room there is for programs that
compute other partial functions, implying that the smallest programs for each
distinct partial function must be distributed more sparsely.
``Dense'' was used for domain-frequent in~\cite{ScJ99}.

Any frequent programming language is obviously domain-frequent but not
necessarily vice versa.
On the other hand, even if a theorem in this field mentions frequency as an
assumption, the odds are that its proof goes through with domain-frequency.
Whether a real-life programming language such as C++ is domain-frequent, is
surprisingly difficult to find out.
We will discuss this question briefly in Section~\ref{S:frequent}.

\newcommand{\BF}{\textsf{BF}}
As an example of a frequent programming language, \BF\ was mentioned
in~\cite{KSZ05}.
Its full name starts with ``brain'' and then contains a word that is widely
considered inappropriate language, so we follow the convention
of~\cite{KSZ05} and call it \BF.
Information on it can be found on Wikipedia under its real name.
It is an exceptionally simple programming language suitable for recreational
and illustrational but not for real-life programming purposes.
In essence, \BF\ programs describe Turing machines with a read-only input
tape, write-only output tape, and one work tape.
The alphabet of each tape is the set of 8-bit bytes.
However, \BF\ programs only use eight characters.

As a side issue, a non-trivial proof was given in~\cite{KSZ05} that only a
vanishing proportion of character strings over the eight characters are \BF\
programs.
That is, $\lim_{n \to \infty} p(n)/8^n$ exists and is $0$.
It trivially follows that if all character strings over the 8 characters are
considered as instances and failure to compile is considered as non-halting,
then the proportion of hard instances vanishes as $n$ grows.

The only possible compile-time error in \BF\ is that the square brackets
\texttt{[} and \texttt{]} do not match.
Most, if not all, real-life programming languages have parentheses or brackets
that must match.
So it seems likely that compile-time errors dominate also in the case of most,
if not all, real-life programming languages.
Unfortunately, this is difficult to check rigorously, because the syntax and
other compile-time rules of real-life programming languages are complicated.
Using another, simpler line of argument, we will prove the result for both C++
and \BF\ in Section~\ref{S:syntax}.

In any event, if the proportion of hard instances among all character strings
vanishes because the proportion of programs vanishes, that is yet another
example of an anomaly stealing the result.
It is uninteresting in itself, but it rules out the possibility of interesting
results about the proportion of hard instances of size $n$ among all character
strings of size $n$.
Therefore, from now on, excluding Section~\ref{S:syntax}, we focus on the
proportion of hard instances among all programs or program--input pairs.

In the case of program--input pairs, the results may be sensitive to how the
program and its input are combined into a single string that is used as the
input of the tester.
To avoid anomalous results, it was assumed in~\cite{KSZ05,ScJ99} that this
``pairing function'' has a certain property called ``pair-fair''.
The commonly used function $x + (x+y)(x+y+1)/2$ is pair-fair.
To use this pairing function, strings are mapped to numbers and back via their
indices in the shortlex ordering of all finite character strings.

A proof was sketched in~\cite{ScJ99} that, assuming domain-frequency and
pair-fairness, 
$$\forall T \in \apprx{G}: \exists c_T > 0: \exists n_T \in \mathbb{N}:
\forall n \geq n_T: \frac{\hh(n)+\hd(n)}{p(n)} \geq c_T \textrm{ .}$$
That is, the proportion of wrong answers does not vanish.
However, this leaves open the possibility that for any failure rate $c > 0$,
there is a tester that fares better than that for all big enough $n$.
This possibility was ruled out in~\cite{KSZ05}, assuming frequency and
pair-fairness.
(It is probably not important that frequency instead of domain-frequency was
assumed.)
That is, there is a positive constant such that for any tester, the proportion
of wrong answers exceeds the constant for infinitely many sizes of instances:

\begin{equation}\label{E:KSZ05}\!\!
\exists c > 0: \forall T \in \apprx{G}: \forall n_0 \in \mathbb{N}: \exists
n \geq n_0: \frac{\hh(n)+\hd(n)}{p(n)} \geq c \textrm{ .}
\end{equation}
The third main result in~\cite{KSZ05}, adapted and generalized to the present
setting, is the following.
We present its proof to obtain the generalization and to add a detail that the
proof in~\cite{KSZ05} lacks, that is, how $T_{i,j}$ is made to halt for
``wrong sizes''.
Generic-case testers are not mentioned, because Proposition~\ref{P:gen->gen}
gave a related result for them.

\begin{theorem}\label{T:KSZ05-3}
For each programming model and variant E, S, G of the halting problem,
$$
\forall c > 0: \exists T_c \in \apprx{X}~~~~\:\,: \forall n_0 \in \mathbb{N}:
\exists n \geq n_0: \frac{\hhc{c}(n)}{p(n)} \leq c \wedge
\frac{\hdc{c}(n)}{p(n)} = 0 \textrm{ and}
$$
$$
\forall c > 0: \exists T_c \in \threeway{X}: \forall n_0 \in \mathbb{N}:
\exists n \geq n_0: \frac{\hhc{c}(n)}{p(n)} \leq c\textrm{ .}\hspace{26mm}
$$
\end{theorem}
\begin{proof}
Let $C = \lceil 1/c \rceil$.
Consider the family $T_{i,j}$ of the programs of the following kind, where $i
\in \mathbb{N}$, $j \in \mathbb{N}$, and $0 \leq i \leq C$.
If $n < j$, $T_{i,j}$ answers ``no'' in the case of approximating and ``I
don't know'' in the case of three-way testers.
If $n \geq j$, $T_{i,j}$ simulates all instances of size $n$ until $\lceil i
p(n) / C \rceil$ of them have halted.
If the simulation stage terminates, then if the given instance is among those
that halted, $T_{i,j}$ answers ``yes'', otherwise $T_{i,j}$ answers ``no'' or
``I don't know''.
Thus an approximating $T_{i,j}$ has $\hdc{i,j}(n) = 0$.

We prove next that some $T_{i,j}$ is the required tester.
Let $i_n = \lfloor C h(n) / p(n) \rfloor$.
Then $i_n p(n) / C \leq h(n) < (i_n+1) p(n) / C$.
When $n \geq j$, the simulation stage of $T_{i_n,j}$ terminates and the
proportion of hard halting instances of $T_{i_n,j}$ is less than $1 / C \leq
c$.
Some $0 \leq i \leq C$ is the $i_n$ for infinitely many values of $n$.
Furthermore, there is a smallest such $i$.
We denote it with $i'$.
There also is a $j$ such that when $n \geq j$, then $i_n \geq i'$.
With these choices, $T_{i',j}$ always halts.
\end{proof}
For a small enough $c$ and the approximating tester $T_c$ in
Theorem~\ref{T:KSZ05-3},~(\ref{E:KSZ05}) implies that the failure rate of
$T_c$ oscillates, that is, does not approach any limit as $n \to \infty$.
This observation is directly obtainable from Lemma~23 in~\cite{KSZ05}.

\subsection{Results on Turing Machines}

For Turing machines with one-way infinite tape and randomly chosen transition
function, the probability of falling off the left end of the tape before
halting or repeating a state approaches $1$ as the number of states
grows~\cite{HaM06}.
The tester simulates the machine until it falls off the left end, halts, or
repeats a state.
If falling off the left end is considered as halting, then the proportion of
hard instances vanishes as the size of the machine grows.
This can be thought of as yet another example of an anomaly stealing the
result.

Formally, $\exists T \in \threeway{X}: \lim_{n \to \infty} (\hh(n)+\hd(n)) /
p(n) = 0$, that is,
$$\exists T \in \threeway{X}: \forall c > 0: \exists n_c \in \mathbb{N}:
\forall n \geq n_c: \frac{\hh(n)+\hd(n)}{p(n)} \leq c\textrm{ .}$$
Here X may be E, S, or G.
Although E was considered in~\cite{HaM06}, the proof also applies to S and G.
Comparing the result to Theorem~\ref{T:itself} in Section~\ref{S:frequent}
reveals that the representation of programs as transition functions of Turing
machines is not domain-frequent.

On the other hand, independently of the tape model, the proportion does not
vanish exponentially fast~\cite{Ryb07}.
Like in~\cite{HaM06}, the proportion is computed on the transition functions,
and not on some textual representations of the programs.
The proof relies on the fact that any Turing machine has many obviously
similarly behaving copies of bigger and bigger sizes.
They are obtained by adding new states and transitions while keeping the
original states and transitions intact.
So the new states and transitions are unreachable.
They are analogous to dead code.
These copies are not common enough to satisfy Definition~\ref{D:frequent}, but
they are common enough to rule out exponentially fast vanishing.
Generic-case decidability was used in~\cite{Ryb07}, but the result applies
also to three-way testers by Proposition~\ref{P:3way->}.

The results in~\cite{CaS08} are based on using weighted running times.
For every positive integer $k$, the proportion of halting programs that do not
halt within time $k+c$ is less than $2^{-k}$, simply because the proportion of
times greater than $k+c$ is less than $2^{-k}$.
The publication presents such a weighting that $c$ is a computable constant.

Assume that programs are represented as self-delimiting bit strings on the
input tape of a universal Turing machine.
The smallest three-way tester of variant E that answers ``yes'' or ``no'' up
to size $n$ and ``I don't know'' for bigger programs, is of size $n \pm
O(1)$~\cite{Val12}.

\section{Programming Languages with Assumptions}\label{S:proglang}

\subsection{Domain-Frequent Languages}\label{S:frequent}

The assumption that the programming language is domain-frequent
(Definition~\ref{D:frequent}) makes it possible to use a small variation of
the standard proof of the non-existence of halting testers, to prove that each
halting tester of variant S has a non-vanishing set of hard instances.
For three-way and generic-case testers, one can also say something about
whether the hard instances are halting or not.
Despite its simplicity, as far as we know, the following result has not been
presented in the literature.
However,~see the comment on~\cite{ScJ99} in Section~\ref{S:KSZ05}.

\begin{theorem}\label{T:itself}
If the programming language is domain-frequent, then
$$
\forall T \in \threeway{S}: \exists c_T > 0: \exists n_T \in \mathbb{N}:
\forall n \geq n_T: \frac{\hh(n)}{p(n)} \geq c_T \wedge \frac{\hd(n)}{p(n)}
\geq c_T \textrm{ ,}
$$
$$
\forall T \in \generic{S}: \exists c_T > 0: \exists n_T \in \mathbb{N}:
\forall n \geq n_T: \frac{\hd(n)}{p(n)} \geq c_T \textrm{ , and\hspace{11mm}}
$$
$$
\forall T \in \apprx{S}: \exists c_T > 0: \exists n_T \in \mathbb{N}:
\forall n \geq n_T: \frac{\hh(n)+\hd(n)}{p(n)} \geq c_T
\textrm{ .\hspace{4mm}}
$$
\end{theorem}
\begin{proof}
Let the execution of $X$ with an input $y$ be denoted with $X(y)$.
For any $T$, consider the program $P_T$ that first tries its input $x$ with
$T$.
If $T(x)$ replies ``yes'', then $P_T(x)$ enters an eternal loop.
If $T(x)$ replies ``no'', then $P_T(x)$ halts immediately.
The case that $T(x)$ replies ``I don't know'' is discussed below.
If $T(x)$ fails to halt, then $P_T(x)$ cannot continue and thus also fails to
halt.

By the definition of domain-frequent, there are $c_T > 0$ and $n_T \in
\mathbb{N}$ such that when $n \geq n_T$, at least $c_T p(n)$ programs halt on
precisely the same inputs as $P_T$.
Let $P'$ be any such program.
Consider $P_T(P')$.
If $T(P')$ answers ``yes'', then $P_T(P')$ fails to halt.
Then also $P'(P')$ fails to halt.
Thus ``yes'' cannot be the correct answer for $T(P')$.
A similar reasoning reveals that also ``no'' cannot be the correct answer for
$T(P')$.
So $P'$ is a hard instance for $T$.

Nothing more is needed to prove the claim for approximating testers.
In the case of generic-case testers, the hard instances make $T$ and thus
$P_T$ fail to halt, so they are non-halting instances.

In the case of three-way testers, all hard instances can be made halting
instances by making $P_T$ halt when $T$ replies ``I don't know''.
This proves the claim $\hh(n)/p(n) \geq c_T$.
The claim $\hd(n)/p(n) \geq c_T$ is proven by making $P_T$ enter an eternal
loop when $T$ replies ``I don't know''.
These two proofs may yield different $c_T$ values, but the smaller one of them
is suitable for both.
Similarly, the bigger of their $n_T$ values is suitable for both.
\end{proof}

The second claim of Theorem~\ref{T:itself} lacks a $\hh(n)$ part.
Indeed, Proposition~\ref{P:gen->gen} says that with generic-case testers,
$\hh(n)$ can be made $0$.
With approximating testers, $\hh(n)$ can be made $0$ at the cost of $\hd(n)$
becoming $d(n)$, by always replying ``yes''.
Similarly, $\hd(n)$ can be made $0$ by always replying ``no''.

The next theorem applies to testers of variant E and presents some results
similar to Theorem~\ref{T:itself}.
To our knowledge, it is the first theorem of its kind that applies to the
halting problem on the empty input.
It assumes not only that many enough equivalent copies exist but also that
they can be constructed.
On the other hand, its equivalence only pays attention to the empty input.

\begin{definition}\label{D:cdf}
A programming language is \emph{computably empty-frequent} if and only if
there is a decidable equivalence relation ``$\,\approx$'' between programs
such that
\begin{itemize}
\item for each program $\pi$, there are $c_\pi > 0$ and $n_\pi \in \mathbb{N}$
such that for every $n \geq n_\pi$, at least $c_\pi p(n)$ programs of size $n$
are equivalent to $\pi$, and
\item for each programs $\pi$ and $\pi'$, if $\pi \approx \pi'$, then either
both or none of $\pi$ and $\pi'$ halt on the empty input.
\end{itemize}
If $\pi \approx \pi'$, we say that $\pi'$ is a \emph{cousin} of $\pi$.
\end{definition}
It can be easily seen from~\cite{KSZ05} that \BF\ is computably
empty-frequent.

\begin{theorem}\label{T:ill-d-R1}
If the programming language is computably empty-frequent, then
$$\forall T \in \threeway{E}: \exists c_T > 0: \exists n_T \in \mathbb{N}:
\forall n \geq n_T: \frac{\hd(n)}{p(n)} \geq c_T\textrm{ .}$$
The result also holds for generic-case testers but not for approximating
testers.
\end{theorem}
\begin{proof}
Given any three-way tester $T$, consider a program $P_T$ that behaves as
follows.
First it constructs its own code and stores it in a string variable.
Hard-wiring the code of a program inside the program is somewhat tricky, but
it is well known that it can be done.
With G\"odel numberings, the same can be obtained with Kleene's second
recursion theorem.

Then $P_T$ starts constructing its cousins of all sizes and tests each of them
with $T$.
By the assumption, there are $c_T > 0$ and $n_T \in \mathbb{N}$ such that for
every $n \geq n_T$, $P_T$ has at least $c_Tp(n)$ cousins of size $n$.
If $T$ ever replies ``yes'', then $P_T$ enters an eternal loop and thus does
not continue testing its cousins.
If $T$ ever replies ``no'', then $P_T$ halts immediately.
If $T$ replies ``I don't know'', then $P_T$ tries the next cousin.

If $T$ ever replies ``yes'', then $P_T$ fails to halt on the empty input.
By definition, also the tested cousin fails to halt on the empty input.
So the answer ``yes'' would be incorrect.
Similarly, if $T$ ever replies ``no'', that would be incorrect.
So $T$ must reply ``I don't know'' for all cousins of $P_T$.
They are thus hard instances for $T$.
Because there are infinitely many of them, $P_T$ does not halt, so they are
non-halting.

To prove the result for generic-case testers, it suffices to run the tests of
the cousins in parallel, that is, go around a loop where each test that has
been started is executed one step and the next test is started.
If any test ever replies ``yes'' or ``no'', $P_T$ aborts all tests that it has
started and then does the opposite of the reply.

A program that always replies ``no'' is an approximating tester with $\hd(n)
= 0$ for every $n \in \mathbb{N}$.
\end{proof}
The results in this section and Section~\ref{S:KSZ05} motivate the question:
are real-life programming languages domain-frequent?
For instance, is C++ domain-frequent?
Unfortunately, we have not been able to answer it.
We try now to illustrate why it is difficult.

Given any C++ program, it is easy to construct many longer programs that
behave in precisely the same way, by adding space characters, line feeds
(denoted with $\nl$), comments, or dead code such as
\texttt{if(0!=0)}\{\ldots\}.
It is, however, hard to verify that many enough programs are obtained in this
way.
For instance, it might seem that many enough programs can be constructed with
string literals.
We now provide evidence that suggests (but does not prove) that it fails.

Any program of size $n$ can be converted to $(|\Sigma|-3)^k$ identically
behaving programs of size $n+k+12$ by adding \{\texttt{char*s="$\sigma$";}\}
to the beginning of some function, where $\sigma \in (\Sigma \setminus
\{\texttt{"}, \texttt{\char`\\}, \nl \})^k$.
(The purpose of \{ and \} is to hide the variable \texttt{s}, so that it does
not collide with any other variable with the same name.)
More programs are obtained by including escape codes such as
\texttt{\char`\\"} to $\sigma$.

However, it seems that this is a vanishing instead of at least a positive
constant proportion when $k \to \infty$.
In the absence of escape codes, it certainly is a vanishing proportion.
This is because one can add \{\texttt{char*s="$\sigma$",*t="$\rho$";}\}
instead, where $|\sigma| + |\rho| = k-6$.
Without escape codes, this yields $(k-5)(|\Sigma|-3)^{k-6}$ programs.
When $k \to \infty$, $(|\Sigma|-3)^k / ((k-5)(|\Sigma|-3)^{k-6}) =
(|\Sigma|-3)^6 / (k-5) \to 0$.

That is, although string literals can represent information rather densely,
they do not constitute the densest possible way of packing information into a
C++ program (assuming the absence of escape codes).
A pair of string literals yields an asymptotically strictly denser packing.
Similarly, a triple of string literals is denser still, and so on.
Counting the programs in the presence of escape codes is too difficult, but it
seems likely that the phenomenon remains the same.

So string literals do not yield many enough programs.
It seems difficult to first find a construct that does yield many enough
programs, and then prove that it works.

\subsection{End-of-file Data Segment Languages}\label{S:eof-data}

In this section we prove a theorem that resembles Theorem~\ref{T:ill-d-R1},
but relies on different assumptions and has a different proof.

We say that a three-way tester is \emph{$n$-perfect} if and only if it does
not answer ``I don't know'' when the size of the instance is at most $n$.
The following lemma is adapted from~\cite{Val12}.

\begin{lemma}\label{L:n-O(1)}
Each programming language has a constant $e$ such that the size of each
$n$-perfect three-way tester of variant E or S is at least $n-e$.
\end{lemma}
\begin{proof}
Let $T_n$ be any $n$-perfect three-way tester of variant E or S.
Consider a program $P$ that constructs character strings $x$ in shortlex order
and tests them with $T_n$ until $T_n(x)$ replies ``I don't know''.
If $T_n(x)$ replies ``yes'', $P$ simulates $x$ before trying the next
character string.
When simulating $x$, $P$ gives it the empty input in the case of variant E and
$x$ as the input in the case of S.
The reply ``I don't know'' eventually comes, because otherwise $T_n$ would be
a true halting tester.
As a consequence, $P$ eventually halts.
Before halting, $P$ simulates at least all halting programs of size at most
$n$.

The time consumption of any simulated execution is at least the same as the
time consumption of the corresponding genuine execution.
So the execution of $P$ cannot contain properly a simulated execution of $P$.
$P$ does not read any input, so it does not matter whether it is given itself
or the empty string as its input. 
Therefore, the size of $P$ is bigger than $n$.
Because the only part of $P$ that depends on $n$ is $T_n$, there is a constant
$e$ such that the size of $T_n$ is at least $n-e$.
\end{proof}

In any everyday programming language, space characters can be added freely
between tokens.
Motivated by this, we define that a \emph{blank character} is a character
that, for any program, can be added to at least one place in the program
without affecting the meaning of the program.

\begin{theorem}\label{T:ill-eof}
Let X be E or S.
If the programming language is end-of-file data segment and has a blank
character, then
$$\forall T \in \threeway{X}: \exists c_T > 0: \exists n_T \in \mathbb{N}:
\forall n \geq n_T: \frac{\hh(n)}{p(n)} \geq c_T \wedge \frac{\hd(n)}{p(n)}
\geq c_T\textrm{ .}$$
\end{theorem}
\begin{proof}
Assume first that tester $T$ is a counter-example to the $\hh$-claim.
That is, for every $c > 0$, $T$ has infinitely many values of $n$ such that
$\hh(n)/p(n) < c$.

If $T$ uses its data segment, let the use be replaced by the use of ordinary
constants, liberating the data segment for the use described in the sequel.
Let $T_{k,m}$ be the following program.
Here $k$ is a constant inside $T_{k,m}$ represented by $\Theta(\log k)$
characters, and $m$ is the content of the data segment of $T_{k,m}$
interpreted as a natural number $m$ in base $|\Sigma|$.
Let $a$ and $d$ be the sizes of the actual program and data segment of
$T_{k,m}$.
We have $a = \Theta(\log k)$.
Let $x$ be the input of $T_{k,m}$.

The program $T_{k,m}$ first computes $n$ := $k + d$.
If $|x| < n$, then $T_{k,m}$ adds blank characters to $x$, to make its size
$n$.
Next, if $|x| > n$, then $T_{k,m}$ replies ``I don't know'' and halts.
Otherwise $T_{k,m}$ gives $x$ (which is now of size precisely $n$) to $T$.
If $T(x)$ replies ``yes'' or ``no'', then $T_{k,m}$ gives the reply as its
own reply and halts.
Otherwise $T_{k,m}$ constructs each character string $y$ of size $n$ and tests
it with $T$.
$T_{k,m}$ simulates in parallel those $y$ for which $T(y)$ returns ``I don't
know'' until $m$ of them have halted (with $y$ or the empty string as the
input, as appropriate).
Then it aborts those that have not halted.
If $x$ is among those that halted, then $T_{k,m}$ replies ``yes'', otherwise
$T_{k,m}$ replies ``no''.

For each $k \in \mathbb{N}$, there are infinitely many values of $n$ such that
$\hh(n)/p(n) < |\Sigma|^{-k}$.
For any such $n$ we have $\hh(n) < p(n)|\Sigma|^{-k} \leq
|\Sigma|^n|\Sigma|^{-k}$.
So $n-k$ characters suffice for representing $\hh(n)$.
Therefore, there is $T_{k,m}$ such that $d = n-k$ and $m = \hh(n)$.
It is an $n$-perfect three-way tester of size $a + d = d + \Theta(\log k) = n
- k + \Theta(\log k)$.
A big enough $k$ yields a contradiction with Lemma~\ref{L:n-O(1)}.

The proof of the $\hd$-claim is otherwise similar, but $T_{k,m}$ counts the
number $v$ of those $y$ for which $T(y)$ returns ``I don't know'', and
simulates the $y$ until $v-m$ of them have halted.
The $\hh$-claim and $\hd$-claim are combined into a single claim by choosing
the smaller $c_T$ and bigger $n_T$ provided by their proofs.
\end{proof}

\subsection{End-of-file Dead Segment Languages}\label{S:eof-dead}

In this section we show that if dead information can be added extensively
enough, a tester of variant E with an arbitrarily small positive failure rate
exists, but the opposite holds for variant S.
The reason for the result on variant E is that as the size of the programs
grows, a bigger and bigger proportion of programs consists of copies of
smaller programs.
This phenomenon is so strong that to obtain the desired failure rate, it
suffices to know the empty-input behaviour of all programs up to a sufficient
size.

An \emph{end-of-file dead segment language} is defined otherwise like
end-of-file data segment language (Definition~\ref{D:eof-data}), but the
actual program cannot read the data segment.
This is the situation with any self-delimiting real-life programming language,
whose compiler stops reading its input when it has read a complete program.
Any end-of-file dead segment language is frequent and computationally
domain-frequent.

\begin{theorem}\label{T:easy-sd}
For each end-of-file dead segment language,
$$\forall c > 0: \exists T_c \in \threeway{E}: \forall n \in \mathbb{N}:
\frac{\hhc{c}(n)+\hdc{c}(n)}{p(n)} \leq c\textrm{ .}$$
The result also holds with approximating and generic-case testers.
\end{theorem}
\begin{proof}
Let $r(n)$ denote the number of programs whose dead segment is not empty.
We have $r(n) \leq p(n) \leq |\Sigma|^n$, so $r(n)|\Sigma|^{-n} \leq 1$.
For each $n \in \mathbb{N}$, $r(n+1) = |\Sigma|p(n) \geq |\Sigma|r(n)$.
So $r(n)|\Sigma|^{-n}$ grows as $n$ grows.
These imply that there is $\ell$ such that $r(n)|\Sigma|^{-n} \to \ell$ from
below when $n \to \infty$.

Because there are programs, $\ell > 0$.
For every $c > 0$ we have $\ell c > 0$, so there is $n_c \in \mathbb{N}$ such
that $r(n_c) |\Sigma|^{-n_c} \geq \ell - \ell c$.
On the other hand, $p(n) = r(n+1)/|\Sigma| \leq \ell|\Sigma|^n$.

These imply $p(n_c-1)|\Sigma|^{n-n_c+1} / p(n) = r(n_c)|\Sigma|^{n-n_c} / p(n)
\geq 1 - c$.
Here $p(n_c-1)|\Sigma|^{n-n_c+1}$ is the number of those programs of size $n$
whose actual program is of size less than $n_c$.

The behaviour of a program on the empty input only depends on its actual
program.
Let $n_a$ be the size of the actual program.
Consider a three-way tester that looks the answer from a look-up table if $n_a
< n_c$ and replies ``I don't know'' if $n_a \geq n_c$ (cf.\
Proposition~\ref{P:finite}).
It has $(\eh(n) + \ed(n)) / p(n) \geq 1-c$, implying the claim.

Proposition~\ref{P:3way->} generalizes the result to approximating and
generic-case testers.
\end{proof}

The above proof exploited the fact that the correct answer for a long program
is the same as the correct answer for a similarly behaving short program.
This does not work for testers of variant S, because the short and long
program no longer get the same input, since each one gets itself as its input.
Although the program does not have direct access to its dead segment, it gets
it via the input.
This changes the situation to the opposite of the previous theorem.

\begin{theorem}\label{T:hard-sd}
For each end-of-file dead segment language,
$$\exists c > 0: \forall T \in \threeway{S}: \forall n_0 \in \mathbb{N}:
\exists n \geq n_0: \frac{\hh(n)}{p(n)} \geq c \wedge \frac{\hd(n)}{p(n)} \geq
c\textrm{ ,}$$
$$\exists c > 0: \forall T \in \makebox[20mm][l]{\generic{S}}: \forall n_0 \in
\mathbb{N}:
\exists n \geq n_0: \frac{\hd(n)}{p(n)} \geq c\textrm{ , and}\hspace{13mm}$$
$$\exists c > 0: \forall T \in \makebox[20mm][l]{\apprx{S}}: \forall n_0 \in
\mathbb{N}: \exists n \geq n_0: \frac{\hh(n)+\hd(n)}{p(n)} \geq c\textrm{
.}\hspace{7mm}$$
\end{theorem}
\begin{proof}
We prove first the claims on three-way and generic-case testers.

Let us recall the overall idea of the proof of Theorem~\ref{T:itself}.
In that proof, for any tester $T$, a program $P_T$ was constructed that gives
its input $x$ to $T$.
If $T(x)$ replies ``yes'', then $P_T(x)$ enters an eternal loop.
If $T(x)$ replies ``no'', then $P_T(x)$ halts immediately.
To prove that a three-way tester has many hard (a) halting (b) non-halting
instances, in the case of the ``I don't know'' reply, $P_T(x)$ was made to (a)
halt immediately (b) enter an eternal loop.
All programs that halt on the same inputs as $P_T$ were shown to be hard
instances for $T$.
For each $n$ that is greater than a threshold that may depend on $T$, the
existence of at least $c_Tp(n)$ such programs was proven, where $c_T$ may
depend on $T$ but not on $n$.

We now apply the same idea, but, to get a result where the same constant $c$
applies to all testers $T$, we no longer construct a separate program $P_T$
for each $T$.
Instead, we construct a single program $P$, which obtains $T$ from the size of
the input of $P$.
(A similar idea appears in~\cite{KSZ05}.)
To discuss this, for any $i > 0$, let $P_i$ be the program whose shortlex
index is $i$.
Let $\delta(i) = i-s(i)+1$, where $s(i)$ is the biggest square number that is
at most $i$.
The essence of $\delta(i)$ is that as $i$ gets the values $1$, $2$, $3$,
\ldots, $\delta(i)$ gets each value $1$, $2$, $3$, \ldots\ infinitely many
times.

One more idea needs to be explained before discussing the details of $P$.
Let $\Sigma$ be partitioned to $\Sigma_1$ and $\Sigma_2$ of sizes $\lfloor
\frac{|\Sigma|}{2} \rfloor$ and $\lceil \frac{|\Sigma|}{2} \rceil$.
Let $n_a$ be the size of the actual program of $P$.
For each $n > n_a$, by modifying the dead segment, $|\Sigma|^{n-n_a}$ programs
are obtained that have the same actual program as $P$.
For $i \in \{1,2\}$, let $\Pi_i$ be the set of those of them whose dead
segment ends with a character in $\Sigma_i$.
We have $\frac{1}{3}|\Sigma|^{n-n_a} \leq |\Pi_1| \leq |\Pi_2|$.
Because $0 < p(n) \leq |\Sigma|^n$, by choosing $c =
\frac{1}{3}|\Sigma|^{-n_a}$ we get $\frac{1}{3}|\Sigma|^{n-n_a} / p(n) \geq
c$.

The program $P$ first checks that its input $x$ is a program with a non-empty
dead segment.
If it is not, then $P$ halts immediately.
Otherwise, $P$ constructs $P_{\delta(|x|)}$ by going through all character
strings in the shortlex order until $\delta(|x|)$ programs have been found.
Then $P$ constructs every program $y$ that has the same size, has the same
actual program, and belongs to the same $\Pi_i$ as $x$.
Then $P$ executes the $P_{\delta(|x|)}(y)$ in parallel until any of the
following happens.

If any $P_{\delta(|x|)}(y)$ replies ``yes'', then $P$ enters an eternal loop.
If any $P_{\delta(|x|)}(y)$ replies ``no'', then $P$ aborts the remaining
$P_{\delta(|x|)}(y)$ and halts.
If every $P_{\delta(|x|)}(y)$ replies ``I don't know'', then $P$ halts if $x
\in \Pi_1$, and enters an eternal loop if $x \in \Pi_2$.
If none of the above ever happens, then $P$ fails to halt.

Recall that $n_a$ is the size of the actual program of $P$.
For any tester $T$, there are infinitely many $n$ such that $n > n_a$ and
$P_{\delta(n)}$ is $T$.
For any such $n$, there are $|\Sigma|^{n-n_a}$ programs $P'$ of size $n$ that
have the same actual program as $P$.
Let $P''$ be any of them.
The execution of $P(P'')$ starts $P_{\delta(n)}(P')$ for at least
$\frac{1}{3}|\Sigma|^{n-n_a}$ distinct $P'$.
If $P_{\delta(n)}(P')$ replies ``yes'', then $T$ claims that $P'(P')$ halts.
Then also $P(P')$ halts, because $P$ halts on the same inputs as $P'$, since
they have the same actual program.
Furthermore, $P(P'')$ halts, because $P$ only looks at the size, actual
program, and $\Pi_i$-class of its input, and $P''$ and $P'$ agree on them.
But the halting of $P(P'')$ is in contradiction with the behaviour of $P$
described above.
Therefore, no $P_{\delta(n)}(P')$ can reply ``yes''.
For a similar reason, none of them replies ``no'' either.

In conclusion, all at least $\frac{1}{3}|\Sigma|^{n-n_a}$ distinct $P'$ are
hard instances for $T$.
If $T$ is a three-way tester, it replies ``I don't know'' for all of them.
Depending on whether $P'' \in \Pi_1$ or $P'' \in \Pi_2$, they are hard halting
or hard non-halting instances.
If $T$ is a generic-case tester, it halts on none of these hard instances.
Therefore, also $P(P'')$ and $P''(P'')$ fail to halt.
So they all are hard non-halting instances.

In the case of approximating testers, $P$ is modified such that it lets all
$P_{\delta(|x|)}(y)$ run into completion and counts the ``yes''- and
``no''-replies that they give.
If the majority of the replies are ``no'', then $P$ halts, otherwise $P$
enters an eternal loop.
For the same reasons as above, $P(P'')$ halts if and only if $P(P')$ halts if
and only if $P'(P')$ halts.
So at least half of the replies are wrong.
\end{proof}

Finally, we prove a corollary of the above theorem that deals with \emph{the
halting problem itself}, not with imperfect testers.
Imperfect testers are used in the proof of the corollary, but not in the
statement of the corollary.

\begin{lemma}\label{L:limit}
Let X be any of E, S, and G, and let $f$ be any total computable function from
natural numbers to integers.
If
$$\exists c > 0: \forall T \in \threeway{X}: \forall n_0 \in \mathbb{N}:
\exists n \geq n_0: \frac{\hh(n)}{p(n)} \geq c\textrm{ ,}$$
then $\displaystyle\lim_{n \to \infty} \frac{h(n)-f(n)}{p(n)}$ does not exist.
\end{lemma}
\begin{proof}
Assume that $\lim_{n \to \infty} (h(n)-f(n))/p(n) = x$ and $c > 0$.
Let $i = \lceil -\log_2 c \rceil$.
There is an $x_i$ of the form $m + \sum_{j=1}^{i+1} b_j 2^{-j}$ such that $m$
is an integer, $b_j \in \{0,1\}$ when $1 \leq j \leq i+1$, and $x_i < x \leq
x_i + 2^{-i-1}$.
There also is $n_0$ such that when $n \geq n_0$, then $x_i \leq
(h(n)-f(n))/p(n) < x_i + 2^{-i}$.

A tester $T$ that disobeys the formula is obtained as follows.
If $n < n_0$, $T$ replies ``I don't know''.
If $n \geq n_0$, $T$ simulates all instances of size $n$ until $\lceil x_i
p(n) \rceil + f(n)$ have halted.
If the given instance is among those that halted, then $T$ replies ``yes'' and
otherwise ``I don't know''.
We have $\hh(n)/p(n) < 2^{-i} \leq c$.
\end{proof}

\begin{corollary}\label{C:nolimit1}
Consider variant S of the halting problem and any end-of-file dead segment
language.
Then $\lim_{n \to \infty} h(n)/p(n)$ does not exist.
\end{corollary}

The proof of Lemma~\ref{L:limit} can be modified to approximating testers with
$(\hh(n) + \hd(n))/p(n) \geq c$.
By~(\ref{E:KSZ05}), the limit fails to exist also in the framework
of~\cite{KSZ05}.

\section{C++ without Comments and with Input}\label{S:C++}

\subsection{The Effect of Compile-Time Errors}\label{S:syntax}

We first show that among all character strings of size $n$, those that are not
C++ programs --- that is, those that yield a compile-time error --- dominate
overwhelmingly, as $n$ grows.
In other words, a random character string is not a C++ program except with
vanishing probability.
The result may seem obvious until one realizes that a C++ program may contain
comments and string literals which may contain almost anything.
We prove the result in a form that also applies to \BF.

C++ is not self-delimiting.
After a complete C++ program, there may be, for instance, definitions of new
functions that are not used by the program.
This is because a C++ program can be compiled in several units, and the
compiler does not check whether the extra functions are needed by another
compilation unit.
Even so, if $\pi$ is a C++ program, then $\pi\texttt{0}$ is definitely not a
C++ program and not even a prefix of a C++ program.
Similarly, if $\pi$ is a \BF\ program, then $\pi\texttt{]}$ is not a prefix of
a \BF\ program.

\begin{proposition}\label{P:syntax}
If for every $\pi \in \Pi$ there is $c \in \Sigma$ such that $\pi c \notin
\Gamma$, then
$$\lim_{n \to \infty} \frac{p(n)}{|\Sigma|^n} = 0\textrm{ .}$$
\end{proposition}
\begin{proof}
Let $q(n) = |\Sigma^n \cap \Gamma|$.
Obviously $0 \leq p(n) \leq q(n) \leq |\Sigma|^n$.

Assume first that for every $\varepsilon > 0$, there is $n_\varepsilon \in
\mathbb{N}$ such that $p(n) / q(n) < \varepsilon$ for every $n \geq
n_\varepsilon$.
Because $0 \leq p(n) / |\Sigma|^n \leq p(n) / q(n)$, we get $p(n) / |\Sigma|^n
\to 0$ as $n \to \infty$.

In the opposite case there is $\varepsilon > 0$ such that $p(n) / q(n) \geq
\varepsilon$ for infinitely many values of $n$.
Let they be $n_1 < n_2 < \ldots$.
Because $\pi c$ is not a prefix of any program, $q(n_i+1) \leq |\Sigma| q(n_i)
- p(n_i) \leq (|\Sigma| - \varepsilon) q(n_i)$.
For the remaining values of $n$, obviously $q(n+1)$ $\leq |\Sigma| q(n)$.
These imply that when $n > n_i$, we have $0 \leq p(n)/|\Sigma|^n \leq
q(n)/|\Sigma|^n \leq q(n_i) / |\Sigma|^{n_i} \leq (1 - \varepsilon /
|\Sigma|)^i \to 0$ when $i \to \infty$, which happens when $n \to \infty$.
\end{proof}
Consider a tester $T$ that replies ``no'' if the compilation fails and ``I
don't know'' otherwise.
If compile-time error is considered as non-halting, then
Proposition~\ref{P:syntax} implies that $\eh(n) \to 0$, $\hh(n) \to 0$,
$\ed(n) \to 1$, and $\hd(n) \to 0$ when $n \to \infty$.
As we pointed out in Section~\ref{S:KSZ05}, this is yet another instance of an
anomaly stealing the result.

\subsection{The C++ Language Model}\label{S:C++model}

The model of computation we study in this section is program--input pairs,
where the programs are written in C++, and the inputs obey the rules stated by
the Linux operating system.
Furthermore, $\Sigma$ is the set of all 8-bit bytes.
To make firm claims about details, it is necessary to fix some language and
operating system.
The validity of the details below has been checked with C++ and Linux.
Most likely many other programming languages and operating systems could have
been used instead.

There are two deviations from the real everyday programming situation.
First, of course, it must be assumed that unbounded memory is available.
Otherwise everything would be decidable.
(However, at any instant of time, only a finite number of bits are in use.)
Second, it is assumed that the programs do not contain comments.
This assumption needs a discussion.

Comments are information that is inside the program but ignored by the
compiler.
They have no effect to the behaviour of the compiled program.
We show next that most long C++ programs consist of a shorter C++ program and
one or more comments.

\begin{lemma}\label{L:clC++}
There are at most $(|\Sigma|-1)^n$ comment-less C++ programs of size $n$.
\end{lemma}
\begin{proof}
Everywhere inside a C++ program excluding comments, it is either the case that
\texttt{@} or the case that the new line character $\nl$ cannot occur next.
That is, for every character string $\alpha$, either $\alpha\texttt{@}$ or
$\alpha\nl$ is not a prefix of any comment-less C++ program.
\end{proof}
(Perhaps surprisingly, there indeed are places that are outside comments and
where any byte except $\nl$ can occur.)

\begin{lemma}\label{L:allC++}
If $n \geq 16$, then there are at least $((|\Sigma|-1)^4 + 1)^{(n-19)/4}$ C++
programs of size $n$.
\end{lemma}
\begin{proof}
Let $A = \Sigma \setminus \{\texttt{*}\}$, and let $m = \lfloor n/4 - 4
\rfloor = {\lceil (n - 19) / 4 \rceil}$.
Consider the character strings of the form
\begin{center}
\texttt{int main()}\{\texttt{/*$\alpha\beta$*/}\}
\end{center}
where $\alpha$ consists of $(n \mod 4)$ space characters and $\beta$ is any
string of the form $\beta_1 \beta_2 \cdots \beta_m$, where $\beta_i \in A^4
\cup \{\texttt{*//*}\}$ for $1 \leq i \leq m$.
Each such string is a syntactically correct C++ program of size $n$.
Their number is $((|\Sigma|-1)^4 + 1)^m \geq ((|\Sigma|-1)^4 + 1)^{(n - 19) /
4}$.
\end{proof}

\begin{corollary}
The proportion of comment-less C++ programs among all C++ programs of size $n$
approaches $0$, when $n \to \infty$.
\end{corollary}
\begin{proof}
Let $s = |\Sigma|-1$.
By Lemmas~\ref{L:clC++} and~\ref{L:allC++}, the proportion is at most\\
$s^n / (s^4 + 1)^{(n - 19) / 4} = s^{19} (s^4 / (s^4 + 1))^{(n - 19) / 4} \to
0$, when $n \to \infty$.
\end{proof}
As a consequence, although comments are irrelevant for the behaviour of
programs, they have a significant effect on the distribution of long C++
programs.
To avoid the risk that they cause yet another anomaly stealing the result, we
restrict ourselves to C++ programs without comments.
This assumption does not restrict the expressive power of the programming
language, but reduces the number of superficially different instances of the
same program.

The input may be any finite string of bytes.
This is how it is in Linux.
Although not all such inputs can be given directly via the keyboard, they can
be given by directing the so-called standard input to come from a file.
There is a separate test construct in C++ for detecting the end of the input,
so the end of the input need not be distinguished by the contents of the
input.
There are $256^n$ different inputs of size $n$.

The sizes of a program and input are the number of bytes in the program and
the number of bytes in the input file.
This is what Linux reports.
The size of an instance is their sum.
Analogously to Section~\ref{S:frequent}, the size of a program is additional
information to the concatenation of the program and the input.
This is ignored by our notion of size.
However, the notion is precisely what programmers mean with the word.
Furthermore, the convention is similar to the convention in ordinary (as
opposed to self-delimiting) Kolmogorov complexity theory~\cite{LiV08}.

\begin{lemma}\label{L:pC++}
With the C++ programming model in Section~\ref{S:C++model}, $p(n) <
|\Sigma|^{n+1}$.
\end{lemma}
\begin{proof}
By Lemma~\ref{L:clC++}, the number of different program--input pairs of size
$n$ is at most
$$\sum_{i=0}^n (|\Sigma|-1)^i |\Sigma|^{n-i} \ =\ |\Sigma|^n \sum_{i=0}^n
\big(\frac{|\Sigma|-1}{|\Sigma|}\big)^i \ <\ |\Sigma|^n \sum_{i=0}^\infty
\big(\frac{|\Sigma|-1}{|\Sigma|}\big)^i \ =\ |\Sigma|^{n+1}\textrm{
.}$$\\[-4.1ex]
\end{proof}

\subsection{Proportions of Hard Instances}\label{S:ill}

The next theorem says that with halting testers of variant G and comment-less
C++, the proportions of hard halting and hard non-halting instances do not
vanish.

\begin{theorem}\label{T:ill-h-C}
With the C++ programming model in Section~\ref{S:C++model},
$$\forall T \in \threeway{G}: \exists c_T > 0: \exists n_T \in \mathbb{N}:
\forall n \geq n_T: \frac{\hh(n)}{p(n)} \geq c_T \wedge \frac{\hd(n)}{p(n)}
\geq c_T\textrm{ .}$$
\end{theorem}
\begin{proof}
We prove first the $\hh(n) / p(n) \geq c_T$ part and then the $\hd(n) / p(n)
\geq c_T$ part.
The results are combined by picking the bigger $n_T$ and the smaller $c_T$.

There is a program $P_T$ that behaves as follows.
First, it gets its own size $n_p$ from a constant in its program code.
The constant uses some characters and thus affects the size of $P_T$.
However, the size of a natural number constant $m$ is $\Theta(\log m)$ and
grows in steps of zero or one as $m$ grows.
Therefore, by starting with $m = 1$ and incrementing it by steps of one, it
eventually catches the size of the program, although also the latter may grow.

Then $P_T$ reads the input, counting the number of the characters that it gets
with $n_i$ and interpreting the string of characters as a natural number $x$
in base $|\Sigma|$.
We have $0 \leq x < |\Sigma|^{n_i}$, and any natural number in this range is
possible.
Let $n = n_p + n_i$.

Next $P_T$ constructs every program--input pair of size $n$ and tests it with
$T$.
In this way $P_T$ gets the number $\eh(n)$ of easy halting pairs of size $n$.

Then $P_T$ constructs again every pair of size $n$.
This time it simulates each of them in parallel until $\eh(n) + x$ of them
have halted.
Then it aborts the rest and halts.
It halts if and only if $\eh(n) + x \leq h(n)$.
(It may be helpful to think of $x$ as a guess of the number of hard halting
pairs.)

Among the pairs of size $n$ is $P_T$ itself with the string that represents
$x$ as the input.
We denote it with $(P_T,x)$.
The time consumption of any simulated execution is at least the same as the
time consumption of the corresponding genuine execution.
So the execution of $(P_T,x)$ cannot contain properly a simulated execution of
$(P_T,x)$.
Therefore, either $(P_T,x)$ does not halt, or the simulated execution of
$(P_T,x)$ is still continuing when $(P_T,x)$ halts.
In the former case, $h(n) < \eh(n) + x$.
In the latter case $(P_T,x)$ is a halting pair but not counted in $\eh(n) +
x$, so $h(n) > \eh(n) + x$.
In both cases, $x \neq h(n) - \eh(n)$.

As a consequence, no natural number less than $|\Sigma|^{n_i}$ is $\hh(n)$.
So $\hh(n) \geq |\Sigma|^{n_i} = |\Sigma|^{n-n_p}$.
By Lemma~\ref{L:pC++}, $p(n) < |\Sigma|^{n+1}$.
So for any $n \geq n_p$, we have $\hh(n)/p(n) > |\Sigma|^{-n_p-1}$.

The proof of the $\hd(n) / p(n) \geq c_T$ part is otherwise similar, except
that $P_T$ continues simulation until $p(n) - \ed(n) - x$ pairs have halted.
(Now $x$ is a guess of $\hd(n)$, yielding a guess of $h(n)$ by subtraction.)
The program $P_T$ gets $p(n)$ by counting the pairs of size $n$ whose program
part is compilable.
It turns out that $p(n) - \ed(n) - x \neq h(n)$, so $x$ cannot be $\hd(n)$,
yielding $\hd(n) \geq |\Sigma|^{n_i}$.
\end{proof}
Next we adapt the second main result in~\cite{KSZ05} to our present setting,
with a somewhat simplified proof and obtaining the result also for three-way
and generic-case testers.

\begin{theorem}\label{T:indep}
With the C++ programming model in Section~\ref{S:C++model},
$$
\exists c > 0: \forall T \in \threeway{G}: \forall n_0 \in \mathbb{N}:
\exists n \geq n_0: \frac{\hh(n)}{p(n)} \geq c \wedge \frac{\hd(n)}{p(n)}
\geq c\textrm{ ,}
$$
$$
\exists c > 0: \forall T \in \makebox[21mm][l]{\generic{G}}: \forall n_0 \in
\mathbb{N}:
\exists n \geq n_0: \frac{\hd(n)}{p(n)} \geq c\textrm{ , and}\hspace{13mm}
$$
$$
\exists c > 0: \forall T \in \makebox[21mm][l]{\apprx{G}}: \forall n_0 \in
\mathbb{N}:
\exists n \geq n_0: \frac{\hh(n)+\hd(n)}{p(n)} \geq c\textrm{ .}\hspace{7mm}
$$
\end{theorem}
\begin{proof}
\newcommand{\last}{\textsf{lb}}
The proof follows the same strategy as the proof of Theorem~\ref{T:hard-sd},
but differs in some technical details.

To prove the claim for three-way testers, for any character string $\alpha$,
let $\last(\alpha) = 0$ if $\alpha$ is the empty string, and otherwise
$\last(\alpha)$ is the value of the least significant bit of the last
character of $\alpha$.
For any character strings $\alpha$ and $\beta$, let $\alpha \simeq \beta$ if
and only if $|\alpha| = |\beta|$ and $\last(\alpha) = \last(\beta)$.
For any size $n$ greater than $0$, ``$\simeq$'' has two equivalence classes,
each containing $|\Sigma|^n/2$ character strings.
For any $i > 0$, let $P_i$ be the program whose shortlex index is $i$.

There is a program $P$ that behaves as follows.
We denote its execution on input $\alpha$ with $P(\alpha)$.
Please observe that if $\alpha \simeq \beta$, then $P(\beta)$ behaves in the
same way as $P(\alpha)$.

First $P(\alpha)$ finds the program $P_{\delta(|\alpha|)}$, where $\delta(i)
= i-s(i)+1$, where $s(i)$ is the biggest square number that is at most $i$.

Then $P(\alpha)$ goes through, in the shortlex order, all $\lceil
|\Sigma|^{|\alpha|}/2 \rceil$ character strings $\beta$ such that $\alpha
\simeq \beta$, until any of the termination conditions mentioned below occurs
or $P(\alpha)$ has gone through all of them.
For each $\beta$, it runs $P_{\delta(|\alpha|)}$ on $\beta$.
We denote this with $P_{\delta(|\alpha|)}(\beta)$.
If $P_{\delta(|\alpha|)}(\beta)$ fails to halt, then $P(\alpha)$ never returns
from it and thus fails to halt.
If $P_{\delta(|\alpha|)}(\beta)$ halts replying ``yes'', then $P(\alpha)$
enters an eternal loop, thus failing to halt.
If $P_{\delta(|\alpha|)}(\beta)$ halts replying ``no'', then $P(\alpha)$ halts
immediately.
If $P_{\delta(|\alpha|)}(\beta)$ halts replying ``I don't know'', then
$P(\alpha)$ tries the next $\beta$.
It is not important what $P(\alpha)$ does if $P_{\delta(|\alpha|)}(\beta)$
halts replying something else.

If $P_{\delta(|\alpha|)}(\beta)$ halted replying ``I don't know'' for every
$\beta$ such that $\alpha \simeq \beta$, then $P(\alpha)$ checks whether
$\last(\alpha) = 0$.
If yes, then $P(\alpha)$ enters an eternal loop, otherwise $P(\alpha)$ halts.

Now let $T(Q,\gamma)$ be any three-way tester that tests whether program $Q$
halts on the input $\gamma$.
How the two components $Q$ and $\gamma$ of the input of $T$ are encoded into
one input string is not important.
There is a program that has $P$ hard-coded into a string constant, inputs
$\beta$, calls $T(P,\beta)$, and gives its reply as its own reply.
Let $i$ be the shortlex index of this program, so the program is $P_i$.

There are infinitely many positive integers $j$ such that $\delta(j) = i$.
Let $j$ be such, and let $\alpha$ be any character string of size $j$.
So $P_{\delta(|\alpha|)}$ is $P_i$.
If, during the execution of $P(\alpha)$, $P_i(\beta)$ ever replies ``yes'' or
``no'', then the same happens during the execution of $P(\beta)$, because
$P(\beta)$ behaves in the same way as $P(\alpha)$ (the fact that $P_i(\beta)$
was called implies $\alpha \simeq \beta$).
But that would be incorrect by the construction of $P$.
Therefore, $T(P,\beta)$ replies ``I don't know'' for every $\beta$ of size
$j$.

As a consequence, $T$ has at least $|\Sigma|^j$ hard instances of size
$|P|+j$.
If $j > 0$, then half of them are halting and the other half non-halting,
thanks to the $\last(\alpha) = 0$ test near the end of $P$.
By Lemma~\ref{L:pC++}, $p(n) < |\Sigma|^{n+1}$.
So if $n = |P|+j > |P|$, then
$$\frac{\hh(n)}{p(n)} \geq \frac{|\Sigma|^{n-|P|}}{2 |\Sigma|^{n+1}} =
\frac{1}{2 |\Sigma|^{|P|+1}} \textrm{ ~ and ~ } \frac{\hd(n)}{p(n)} \geq
\frac{1}{2 |\Sigma|^{|P|+1}} \textrm{ .}$$
The program $P$ does not depend on $n$, so letting $c = 1/(2
|\Sigma|^{|P|+1})$ we have the claim.

The proof for generic-case testers is otherwise similar, but the $\beta$ are
tried in parallel and $T(P,\beta)$ fails to halt for every $\beta$ of size
$j$.
All hard instances are non-halting.
The $P$ for approximating testers lets each $P_{\delta(|\alpha|)}(\beta)$
continue until completion, counts the numbers of the ``yes''- and
``no''-replies they yield, and then does the opposite of the majority of the
replies.
\end{proof}

Application of Lemma~\ref{L:limit} to this result yields the following.

\begin{corollary}\label{C:nolimit2}
With the C++ programming model in Section~\ref{S:C++model}, $\lim_{n \to
\infty} h(n)/p(n)$ does not exist.
\end{corollary}

\section{Conclusions}\label{S:conclusions}

This study did not cover all combinations of a programming model, variant of
the halting problem, and variant of the tester.
So there is a lot of room for future work.

The results highlight what was already known since~\cite{Lyn74}: the
programming model has a significant role.
With some programming models, a phenomenon of secondary interest dominates the
distribution of programs, making hard instances rare.
Such phenomena include compile-time errors and falling off the left end of
the tape of a Turing machine.

Many results were derived using the assumption that information can be packed
very densely in the program or the input file.
Sometimes it was not even necessary to assume that the program could use the
information.
It sufficed that the assumption allowed to make many enough similarly behaving
longer copies of an original program.
Intuition suggests that if the program can access the information, testing
halting is harder than in the opposite case.
A comparison of Theorem~\ref{T:easy-sd} to Theorem~\ref{T:hard-sd} supports
this intuition.

Corollaries~\ref{C:nolimit1} and~\ref{C:nolimit2} and the comment after
Corollary~\ref{C:nolimit1} tell that the proportion of \emph{all} (not just
hard) halting instances has no limit with end-of-file dead segment languages
and variant S of the halting problem, with the C++ model and variant G, and in
the framework of~\cite{KSZ05}.
It must thus oscillate irregularly as the size of the program grows ---
irregularly because of Lemma~\ref{L:limit}.
This is not a property of various notions of imperfect halting testers, but a
property of the halting problem itself.

\subsection*{Acknowledgements}
I thank professor Keijo Ruohonen for helpful discussions, and the anonymous
reviewers of SPLST '13 and Acta Cybernetica for their helpful comments.
The latter pointed out that Proposition~\ref{P:syntax} had been formulated
incorrectly.

\end{document}